\documentclass[11pt,reqno]{amsart}
\usepackage{amssymb,latexsym}
\usepackage[demo]{graphicx}
\usepackage{mathtools}
\usepackage{color}
\usepackage[T1]{fontenc}
\usepackage[hidelinks]{hyperref}
\usepackage{amsmath}
\usepackage{mathdots}
\usepackage{breqn}
\usepackage{tabularx}
\usepackage{tikz}
\usetikzlibrary{arrows.meta,positioning,calc}
\usepackage[toc,page]{appendix}
\usepackage{url}    
\usepackage{breqn}
\usepackage{CJK}
\newcommand{\bburl}[1]{\textcolor{blue}{\url{#1}}}

\usepackage{float}

\calclayout

\makeatletter
\newcommand{\monthyear}[1]{%
  \def\@monthyear{\uppercase{#1}}}
\newcommand{\volnumber}[1]{%
  \def\@volnumber{\uppercase{#1}}}
\AtBeginDocument{%
\def\ps@plain{\ps@empty
  \def\@oddfoot{\@monthyear \hfil \thepage}%
  \def\@evenfoot{\thepage \hfil \@volnumber}}
\def\ps@firstpage{\ps@plain}
\def\ps@headings{\ps@empty
  \def\@evenhead{%
    \setTrue{runhead}%
    \def\thanks{\protect\thanks@warning}%
    \uppercase{\ }\hfil}%
  \def\@oddhead{%
    \setTrue{runhead}%
    \def\thanks{\protect\thanks@warning}%
    \hfill\uppercase{}}%
  \let\@mkboth\markboth
  \def\@evenfoot{%
    \thepage \hfil \@volnumber}%
  \def\@oddfoot{%
    \@monthyear \hfil \thepage}%
  }%
\footskip=25pt
\pagestyle{headings}%
}
\makeatother

\theoremstyle{plain}
\numberwithin{equation}{section}
\newtheorem{thm}{Theorem}[section]
\newtheorem{theorem}[thm]{Theorem}

\newtheorem{definition}[thm]{Definition}

\newtheorem{remark}[thm]{Remark}

\newcommand{\ignore}[1]{}

\begin{document}

\monthyear{August 2021}
\volnumber{August 2021}
\setcounter{page}{1}
\title{Nonexistence of a Universal Algorithm for Traveling Salesman Problems in Constructive Mathematics}

\author{Linglong Dai
}

\begin{abstract}
Proposed initially from a practical circumstance, the traveling salesman problem caught the attention of numerous economists, computer scientists, and mathematicians. These theorists were instead intrigued by seeking a systemic way to find the optimal route. Many attempts have been made along the way and all concluded the nonexistence of a general algorithm that determines optimal solution to all traveling salesman problems alike. In this study, we present proof for the nonexistence of such an algorithm for both asymmetric (with oriented roads) and symmetric (with unoriented roads) traveling salesman problems in the setup of constructive mathematics. 
\end{abstract}

\address{\tiny{George School, Newtown, PA 18940}} \email{dail23@georgeschool.org}

\date{\today}
\thanks{This research was supported by the Ivy Mind Summer Program and the project was created as a result of joint collaboration of Viktor Chernov and Vladimir Chernov.}

\maketitle
\ \\ Keywords: Traveling salesman problem, Constructive mathematics.\\ \

\tableofcontents
\newpage
\section{Introduction}
    The traveling salesman problem (TSP) is a problem of major concern in computer science, economics, and mathematics. It abstracts the realistic situation of traveling between cities as a graph. In the graph, cities are denoted as nodes and roads connecting them as edges with each costing the traveler a certain amount of toll fee. For a long time, theorists have worked on finding systemic ways to obtain the optimal route, following which a salesman could travel to all cities non-repeatedly in a complete tour and return to the original location with a minimum cost. Although various algorithms have been proposed and improved upon one another, a universal algorithm that computes the optimal routes in all traveling salesman problems in finite time seems to be non-existing.  
    
    In this paper, we focus on the problem of the existence of such an algorithm. In particular, we consider such a problem in constructive mathematics. The principle of omniscience grants two possible outcomes: given an algorithm, either optimal solutions could be found in all TSPs or optimal solution(s) computable in finite time does not exist for at least one TSP. We present the validity of the latter in the paper. 
    
    Specifically, we formulate our objective in this paper as to prove the following theorem:
    \begin{theorem}\label{theorem}
    There does not exist an everywhere defined computable algorithm $\Hat{H}$ which determines the optimal route in all traveling salesman problems when costs of roads are constructive real numbers.
    \end{theorem}
    \noindent Furthermore, in section $4$, we limit the type of TSPs to symmetric TSPs, meaning that traveling back and forth between two cities would have the same cost. And we prove that:
    \begin{remark}\label{remark}
    \ref{theorem} holds true for symmetric TSPs.
    \end{remark}
    
    Note that there are inextendable algorithms defined on $\mathbb{N}$ to $[0,1]$. 
    \begin{theorem}[\textbf{Shen $\&$ Vereshchagin \cite{SV}}]\label{SV}
    There exists a computable function that has no total computable extension.
    \end{theorem}
    \noindent The fact that there are partially defined inextendable algorithms is crucial in the proofs.

\section{Definitions}\label{defi}
    Before proving \ref{theorem} and \ref{remark}, we define one key term and two sequences. In this paper, all traveling salesman problems of consideration have toll fees that are constructive real numbers, which we define as the following:
    \begin{definition}[\textbf{Constructive Real Numbers}]\label{CRNs}
    Given two programs $\alpha$ and $\beta$. $\alpha$ generates a sequence of rational numbers $\alpha(n), n\in\mathbb{N}$. $\beta$, the regulator program, generates a sequence of natural numbers $\beta(m)$. The sequence of rational numbers $\alpha(n)$ converges to a constructive real number such that 
    \begin{equation}
        \forall m, i,j>\beta(m), \ \ |\alpha(i) - \alpha(j)| \ \leq \ 2^{-m}.
    \end{equation}
    \end{definition}
   
    In addition, we define two sequences that will be utilized in the remaining sections.
    \begin{definition}[\textbf{Sequence $C$}]
    Given a partially defined computable algorithm $H$ that takes in an input $n$, we define sequence $C$ in the following manner: 
    \[ C_{n,k} = \begin{cases} 
      1 & \text{If by step $k$, $H$ didn't finish working on input $n$ or already gave $1$.} \\
      1-2^{-m} & \text{If by step $k$, $H$ finished working on input $n$ and gave $0$.}\\
   \end{cases}
    \]
    Here $m$ is the step number when $0$ was printed and $m,n,k\in \mathbb{N}$.
    \end{definition}
    \begin{definition}[\textbf{Sequence $D$}]
    In a similar manner, we define a new sequence, $D$:
    \[ D_{n,k} = \begin{cases} 
      1 & \text{If by step $k$, $H$ didn't finish working on input $n$ or already gave $0$.} \\
      1-2^{-m} & \text{If by step $k$, $H$ finished working on input $n$ and gave $1$.}\\
   \end{cases}\]
    Same as in the previous definition, $m$ is the step number when $1$ was printed and $m,n,k\in \mathbb{N}$.
    \end{definition}
    For each fixed $n$, $C_{n,k}$ and $D_{n,k}$ are constructive numbers, which we denote by $C_n,D_n$. In either of the subsequent cases, $C_n$ and $D_n$ are constructive real numbers:
    \begin{enumerate}
        \item If program prints $0$ for $C$ or $1$ for $D$, terms are always equal to $1-2^{-m}$ after or by step $m$, which according to \ref{CRNs}, the corresponding $C_n$ and $D_n$ are constructive real numbers with the standard regulator. 
        \item If program prints $1$ for $C$ or $0$ for $D$ or does not terminate, then terms in the sequence are always equal to $1$, by \ref{CRNs}, they are constructive real numbers with the standard regulator. 
    \end{enumerate}
     
\section{The Case of General TSPs}\label{asy}
    In this section, we present a proof of the nonexistence of an algorithm that always decides the optimal solution to general traveling salesman problems by constructing a particular contradiction. 
    
    To begin, we consider the simplistic instance of $3$ nodes and $6$ one-way roads with numbers representing the tolls. $C_{n}$ and $D_{n}$ indicated in the graph are constructive real number(s) generated by the partially defined algorithm $H$ as defined in \ref{defi}.
    \begin{center}
    \begin{figure}[H]
    \begin{tikzpicture}[scale=1.2]
    \node[shape=circle,draw=black] (A) at (7.5,-1) {$n_1$};
    \node[shape=circle,draw=black] (B) at (11,-1) {$n_2$};
    \node[shape=circle,draw=black] (C) at (9.25,2.5) {$n_3$};
    
    \path [->,color=blue] (A) edge node {$C_{n}$} (B);
    \path [->,color=blue] (B) edge node {1} (C);
    \path [->,color=blue] (C) edge node {1} (A);
    \path [->,bend left=50,color=red] (C) edge node {$D_{n}$} (B);
    \path [->,bend left=50,color=red] (B) edge node {1} (A);
    \path [->,bend left=50,color=red] (A) edge node {1} (C);
    \end{tikzpicture} 
    
    \caption{TSP with $3$ Nodes}
    \label{asyGraph}
    \end{figure}
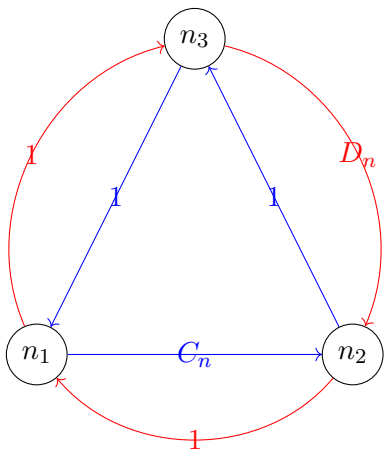

    \end{center}
    In this case, there are two complete tours, $n_1\rightarrow n_2 \rightarrow n_3 \rightarrow n_1$ in blue and $n_1\rightarrow n_3 \rightarrow n_2 \rightarrow n_1$ in red. 
    
    Let an algorithm $\hat{H}$ be the algorithm that solves all traveling salesman problems by deciding the optimal route with a minimum cost. We have a partially defined computable algorithm $H$ that generates sequence $C$ and $D$. There is an extension of algorithm $H$, $H'$, that applies to the aforementioned instance by determining $\min(2+C_n,2+D_n)$, which further reduces to computing $\min(C_n,D_n)$.
    
    To explain the process of determining the optimal route by $H'$, we shall enumerate a few possible $C$ and $D$ and the corresponding outputs of $H'$. We use $Z$ to denote the case when the program does not terminate. 
    \renewcommand{\arraystretch}{1.5}
    \begin{center}
    \begin{table}[H]
    \begin{tabular}{c c c c c c c c c c c c c c c c c}
    $C_n \backslash k$ & & & & & $\text{output of} \atop H$ & & & & $D_n \backslash k$ & & & & & $\text{output of} \atop H$ &  \\ 
    & $\frac{1}{2}$ & $\frac{1}{2}$ & $\frac{1}{2}$ & $\cdots$ & (0) & & & & & $\frac{1}{2}$ & $\frac{1}{2}$ & $\frac{1}{2}$ & $\cdots$ & \ \ \ (1) &\\  
    & 1 & $\frac{1}{4}$ & $\frac{1}{4}$ & $\cdots$ & (0) & & & & & 1 & $\frac{1}{4}$ & $\frac{1}{4}$ & $\cdots$ &  \ \ \ (1) &\\
    & 1 & 1 & $\frac{1}{8}$ & $\cdots$ & (0) & & & & & 1 & 1 & $\frac{1}{8}$ & $\cdots$ &  \ \ \ (1) & \\
    & & & & & $\vdots$ & & & & & & & & &  \ \ \ \ $\vdots$ & \\
     & 1 & 1 & 1 & $\cdots$ & \text{(Z or 1)} & & & & & 1 & 1 & 1 & $\cdots$ & \text{(Z or 0)} & \\
    &&&&&&&&&&&&&&\\
    \end{tabular}
    \caption{Enumerations of possible sequence $C$ and $D$.}
    \label{table}
    \end{table}
    \end{center}
    As we shall observe from the table that each corresponding line of some $C_n$ and some $D_n$ satisfies the definition of constructive real numbers \ref{CRNs}, and are possible tolls of some particular roads. Notice that when $H$ outputs $1$, $D_n\leq C_n$, $H'$ decides that the red route is the optimal solution, and when $H$ outputs $0$, $C_n\leq D_n$, we would otherwise have the blue route as the optimal solution. When $C_n=D_n$, both routes are optimal, which means that the previous statement would also be correct.
    
    We proceed to prove \ref{theorem} by contradiction. 
    \begin{theorem}
    There does not exist an everywhere defined computable algorithm $\hat{H}$ which determines the optimal route in all traveling salesman problems when costs of roads are constructive real numbers.
    \end{theorem}
    \begin{proof}
    Assume there exists an algorithm $\hat{H}$ that always finds the optimal routes in traveling salesman problems. Let algorithm $H$ be a partially defined computable algorithm that generates sequence $C$ and $D$ mentioned in \ref{defi}. Then there exists an extension of $H$, $H'$, that determines either $C_n$ or $D_n$ is smaller, and decides the optimal route in the previously mentioned traveling salesman problem. The extension poses a contradiction to our knowledge about $H$ being an inextendable algorithm (see \ref{SV}). Therefore, we have proven \ref{theorem}. 
    \end{proof}
    
    The idea of this proof could be utilized to prove \ref{remark}, which we present in the next section.

\section{The Case of Symmetric TSPs}
    In section \ref{asy}, we proved \ref{theorem} by constructing a contradiction with an asymmetric traveling salesman problem. In this section, we attempt to limit the range of traveling salesman problems down to the symmetric cases, and we prove the nonexistence of an algorithm that finds the optimal solution in all symmetric traveling salesman problems (with unoriented roads). 
    
    Similarly, we begin by constructing a particular traveling salesman problem: 
    \begin{centering}
    \begin{figure}[H]
    \begin{tikzpicture}[scale = 0.9]
    \node[shape=circle,draw=black] (A) at (6.25,-0.2) {$n_2$};
    \node[shape=circle,draw=black] (B) at (12.25,-0.2) {$n_5$};
    \node[shape=circle,draw=black] (C) at (9.25,2.5) {$n_1$};
    \node[shape=circle,draw=black] (D) at (7.5,-4) {$n_3$};
    \node[shape=circle,draw=black] (E) at (11,-4) {$n_4$};
    
    \draw [<->] (A) -- (B) node[midway]{1};
    \draw [<->] (A) -- (C) node[midway]{1};
    \draw [thick, <->, dotted] (A) -- (D) node[midway]{100};
    \draw [thick, <->, dotted] (A) -- (E) node[midway]{100};
    \draw [thick, <->, dotted] (B) -- (C) node[midway]{100};
    \draw [<->] (B) -- (D) node[midway]{1};
    \draw [<->] (B) -- (E) node[midway]{1};
    \draw [<->] (C) -- (D) node[midway]{$C_n$};
    \draw [<->] (C) -- (E) node[midway]{$D_n$};
    \draw [<->] (D) -- (E) node[midway]{1};
    \end{tikzpicture}
    
    \caption{TSP with $5$ Nodes}
    \label{symGraph}
    \end{figure}
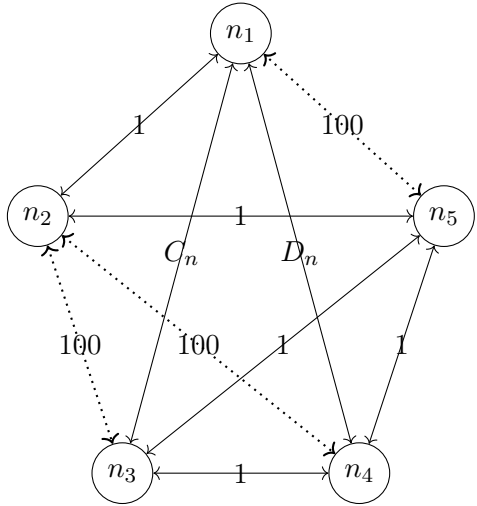
    \end{centering}
    The dotted paths in the diagram \ref{symGraph} have costs of $100$, which are economically inefficient. Therefore, complete tours (tours that connect $5$ nodes and return to the initial position) with at least one road costing $100$ could be ignored in the selection of the optimal route in the course of the proof. In the hypothetical situation of machine decision, the algorithm $H'$ would also eliminate those cases as it recognizes them as economically inefficient when making a comparison between tours marked in the following diagram. The binary search would culminate in comparing the tours below, which are the only candidates of the optimal solution in the aforementioned TSP.
    \begin{figure}[H]
            \begin{center}
    \begin{tikzpicture}[scale=0.8]
    \node[shape=circle,draw=black] (A) at (6.25,-0.2) {$n_2$};
    \node[shape=circle,draw=black] (B) at (12.25,-0.2) {$n_5$};
    \node[shape=circle,draw=black] (C) at (9.25,2.5) {$n_1$};
    \node[shape=circle,draw=black] (D) at (7.5,-4) {$n_3$};
    \node[shape=circle,draw=black] (E) at (11,-4) {$n_4$};
    
    \draw [thick, <->, color = blue] (A) -- (B) node[midway]{1};
    \draw [thick, <->, color = blue] (A) -- (C) node[midway]{1};
    \draw [thick, <->, dotted] (A) -- (D) node[midway]{100};
    \draw [thick, <->, dotted] (A) -- (E) node[midway]{100};
    \draw [thick, <->, dotted] (B) -- (C) node[midway]{100};
    \draw [<->] (B) -- (D) node[midway]{1};
    \draw [thick, <->, color = blue] (B) -- (E) node[midway]{1};
    \draw [thick, <->, color = blue] (C) -- (D) node[midway]{$C_n$};
    \draw [<->] (C) -- (E) node[midway]{$D_n$};
    \draw [thick, <->, color = blue] (D) -- (E) node[midway]{1};
    \end{tikzpicture}
    \hspace{1cm}
    \begin{tikzpicture}[scale=0.8]
    \node[shape=circle,draw=black] (A) at (6.25,-0.2) {$n_2$};
    \node[shape=circle,draw=black] (B) at (12.25,-0.2) {$n_5$};
    \node[shape=circle,draw=black] (C) at (9.25,2.5) {$n_1$};
    \node[shape=circle,draw=black] (D) at (7.5,-4) {$n_3$};
    \node[shape=circle,draw=black] (E) at (11,-4) {$n_4$};
    
    \draw [thick, <->, color = red] (A) -- (B) node[midway]{1};
    \draw [thick, <->, color = red] (A) -- (C) node[midway]{1};
    \draw [thick, <->, dotted] (A) -- (D) node[midway]{100};
    \draw [thick, <->, dotted] (A) -- (E) node[midway]{100};
    \draw [thick, <->, dotted] (B) -- (C) node[midway]{100};
    \draw [thick, <->, color = red] (B) -- (D) node[midway]{1};
    \draw [<->] (B) -- (E) node[midway]{1};
    \draw [<->] (C) -- (D) node[midway]{$C_n$};
    \draw [thick, <->, color = red] (C) -- (E) node[midway]{$D_n$};
    \draw [thick, <->, color = red] (D) -- (E) node[midway]{1};
    \end{tikzpicture}
    \end{center}
    
    \caption{The Only $2$ Possibilities}
    \label{symGraph}
    \end{figure}
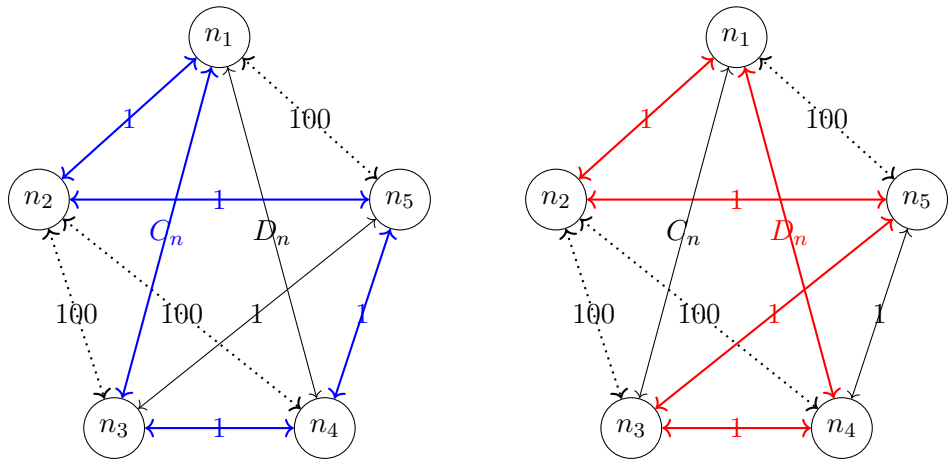
    As indicated in the graphs, the tour in blue, $n_1\rightarrow n_2 \rightarrow n_5 \rightarrow n_4 \rightarrow n_3 \rightarrow n_1$, costs the traveler $4+C_n$; the tour in red, $n_1 \rightarrow n_2 \rightarrow n_5 \rightarrow n_3 \rightarrow n_4 \rightarrow n_1$\footnote{The aforementioned tours are unorientated, $n_1 \rightarrow n_3 \rightarrow n_4 \rightarrow n_5 \rightarrow n_2 \rightarrow n_1$ for example, is considered the same as $n_1\rightarrow n_2 \rightarrow n_5 \rightarrow n_4 \rightarrow n_3 \rightarrow n_1$.}, costs the traveler $4+D_n$. $C_n$ and $D_n$ represent some constructive numbers between $[0,1]$. 
    
    We then follow the idea in the previous section to prove the following: 
    \begin{remark}
    \ref{theorem} holds true for symmetric TSPs.
    \end{remark}
    \begin{proof}
    Assume there is an algorithm $\hat{H}$ that always finds the optimal routes in symmetric traveling salesman problems. Let a partially defined computable algorithm $H$ be the algorithm that generates sequence $C$ and $D$. We could find an extension of $H$, $H'$, that applies specifically to the previously mentioned symmetric traveling salesman problem with $5$ nodes. To determine the optimal route, the algorithm compares $4+C_n$ and $4+D_n$, which reduces to $C_n$ and $D_n$. With the same explanation in \ref{asy}, we would conclude that such an algorithm exists. However, since algorithm $H$ is inextendable (see \ref{SV}), the constructed $H'$ poses a contradiction. Therefore, $H$ does not exist and we proved \ref{remark}.
    \end{proof}
    
\section{Conclusion and Future Work}
    In this paper, we proved the nonexistence of a universal algorithm that determines the optimal complete tour in all traveling salesman problems (both symmetric and asymmetric) in a constructive mathematical setup and further in all symmetric traveling salesman problems. 
    
    For future research, we wish to explore the possibility of creating an infinite series of examples as in \ref{asyGraph} and \ref{symGraph}. Additionally, we propose that gaining topological alternatives to understand and construct examples might be a plausible direction.  



\begin{thebibliography}{2}
\bibitem[BB]{BB} Errett Bishop and Douglas Bridges, \emph{Constructive Analysis}. (1985). Springer-Verlag. 

\bibitem[BGG]{BGG} Egon Börger, Erich Grädel, and Yuri Gurevich, \emph{The Classical Decision Problem}.

\bibitem[SV]{SV}A. Shen and N. K. Vereshchagin, \emph{Computable Functions}. (2002). American Mathematical Society.
\end{thebibliography}
\end{document}